\documentclass[letterpaper, 10 pt, conference]{ieeeconf}  

\IEEEoverridecommandlockouts         
\usepackage{graphicx}

\usepackage{amssymb,amsmath,amsthm}

\usepackage{comment}
\usepackage{amsfonts}
\usepackage{algorithmic}
\usepackage{graphicx}
\usepackage{textcomp}
\usepackage{xcolor}
\usepackage{caption}
\usepackage{subcaption}
\usepackage{bm}
\usepackage{dirtytalk}
\newtheorem{theorem}{Theorem}
\newtheorem{problem}{Problem}
\theoremstyle{definition}
\newtheorem{example}{Example}
   
\DeclareMathOperator*{\argmaxA}{arg\,max} 
\newcommand{\norm}[1]{\left\lVert#1\right\rVert}
\usepackage[utf8]{inputenc}

\newcommand\xqed[1]{%
  \leavevmode\unskip\penalty9999 \hbox{}\nobreak\hfill
  \quad\hbox{#1}}
\newcommand\demo{\xqed{$\square$}}

\IEEEoverridecommandlockouts  
\title{\LARGE \bf
Neural Network-based Control for Multi-Agent Systems from  Spatio-Temporal Specifications}

\author{Suhail Alsalehi, Noushin Mehdipour, Ezio Bartocci and Calin Belta 
\thanks{This work was partially supported by the National Science Foundation under grants IIS-2024606 and IIS-1723995 at Boston University and by the Austrian FFG-funded IoT4CPS project at TU Wien.}
\thanks{Suhail Alsalehi, Noushin Mehdipour and Calin Belta are with the Division of Systems Engineering, Boston University, Boston, MA 02215 USA
        {\tt\small (alsalehi, noushinm, cbelta)@bu.edu}}%
\thanks{Ezio Bartocci is with TU Wien, Vienna, Austria
        {\tt\small ezio.bartocci@tuwien.ac.at}}%
}

\begin{document}

\maketitle
\thispagestyle{empty}
\pagestyle{empty}

\begin{abstract}
We propose a framework for solving control synthesis problems for multi-agent networked systems required to satisfy spatio-temporal specifications. We use Spatio-Temporal Reach and Escape Logic (STREL) as a specification language. For this logic, we define smooth quantitative semantics, which captures the degree of satisfaction of a formula by a multi-agent team. We use the novel quantitative semantics to map control synthesis problems with STREL specifications to optimization problems and propose a combination of heuristic and gradient-based methods to solve such problems. As this method might not meet the requirements of a real-time implementation, we develop
a machine learning technique that uses the results of the off-line optimizations to train a neural network that gives the control inputs at current states. We illustrate the effectiveness of the proposed framework by applying it to a model of a robotic team required to satisfy a spatial-temporal specification under communication constraints. 
\end{abstract}

\section{Introduction}
\label{sec:intro}

Multi-agent systems are used as models in many applications, ranging from robotics to power networks, smart cities, and synthetic biology. Planning and controlling the motions of multi-agent systems are difficult problems, which have received a lot of attention in recent years. One of the main challenges is specifying their motions. Existing approaches include consensus algorithms \cite{bullo2009distributed,mesbahi2010graph}, in which the specification is to reach a desired global state (e.g., minimum / maximum inter-robot separation, specified centroid, heading alignment, etc.) and abstractions, in which a team is parameterized by a set of features (e.g., mean, variance, orientation), \cite{Belta-TRO04,Lynch-swarm}. 

Recently, spatio-temporal logics have emerged as formal ways to specify both spatial and temporal logic requirements for spatially distributed systems \cite{Haghighi2016,Liu2020}. Examples include Spatial Temporal Logic (SpaTeL)~\cite{Haghighi2015}, 
Spatial  Aggregation  Signal  Temporal  Logic
(SaSTL)~\cite{SaSTL}, and Spatio-Temporal Reach and Escape Logic (STREL)~\cite{Bartocci2017,Bartocci2020}. 
SpaTeL is the unification of Signal Temporal Logic (STL)~\cite{STL2004} and the spatial logic Tree Spatial Superposition Logic (TSSL) \cite{TSSL}. Even though it was used as a specification language for a multi-robot system \cite{Haghighi2016}, SpaTel cannot capture the satisfaction for individual agents. Furthermore, TSSL is computationally very expensive. 
SaSTL extends STL with two operators for expressing spatial aggregation and spatial counting, and it was used to monitor  safety and performance requirements of smart cities. 
STREL extends STL with the spatial operators \emph{reach} and \emph{escape}, from which it is possible to derive the spatial modalities \emph{everywhere}, \emph{somewhere}, and \emph{surround}. SaSTL and STREL allow for specifying the requirements for individual agents and can inform about the satisfaction of the specifications locally, as opposed to SpaTeL. 

In this work, we employ STREL to specify complex spatio-temporal requirements for individual agents in multi-agent teams. STREL formulas use Boolean operators, such as conjunction ($\wedge$), disjunction ($\vee$), and negation ($\neg$); temporal operators, such as \emph{eventually} ($F_{[a,b]}$) and  \emph{always} ($G_{[a,b]}$), where $[a,b]$ is a time interval; and spatial operators, such as \emph{reach} ($\mathcal{R}$), \emph{escape} ($\mathcal{E}$), and \emph{surround} ($\odot$). For example, the requirement \say{the agents must always surround the target in the time interval $[a,b]$} can be specified using the STREL formula $\varphi = G_{[a,b]} (agents \odot target)$. The original STREL semantics, as defined in \cite{Bartocci2017}, has no notion of spatial counting, which can be critical in some multi-agent robotic scenarios \cite{Yunus2020}. For instance, one might want to maximize the number of agents that surround a given target. 
Furthermore, the original semantics does not account for the distance between agents, which can be critical to connectivity among agents in the team. To overcome these setbacks, we propose new counting quantitative STREL semantics that allows for spatial counting and depends on the distances among agents.

Similar to STL, SpaTel and STREL are equipped with quantitative semantics, or robustness functions, which quantify the degree of satisfaction of a formula by a (temporal and spatial) trajectory of a system, and allow for mapping control problems into optimization problems. Mixed Integer Linear Program (MILP) encodings were proposed in \cite{Raman2014,Haghighi2016,Liu2018} to solve such problems. Although this method has shown some promising results, MILP encodings are complicated, and the performance times of the MILP solvers are difficult to predict. Gradient-based methods were proposed in \cite{Pant2017,Haghighi2019-as,varnai2020robustness,MehdipourAGM,Gilpin2021}. These 
have fast convergence rates and are simple to implement. However, they are prone to premature convergence and need to be initialized carefully.

To overcome the limitations of the two classes of optimization approaches mentioned above, we propose using a hybrid optimization method that combines heuristic and gradient-based algorithms \cite{mavrovouniotis2017survey,VICTOIRE200451}. The optimization is carried out in two stages. First, the search space is explored using a heuristic algorithm to find a good candidate solution, ideally near the global optima. In the second stage, the best candidate solution from the previous stage is improved by means of a gradient-based algorithm. 
The STREL quantitative semantics is based on 
$min/max$ functions and the robustness function is not differentiable. We use a smooth approximation \cite{Pant2017}, which allows us to employ gradient-based methods for optimization. We show that our approach provides better exploration of the search space and fast convergence times compared to the optimization approaches used in the literature.

Generating control inputs by solving optimization problems as described above is still computationally expensive, and not amenable for real-time control. In this paper, we propose an approach to real-time control using Recurrent Neural Networks (RNN) \cite{liu2020recurrent,yaghoubi2020training}. The RNN learns controllers from state-control trajectories generated off-line by solving the control synthesis problem with different initializations. Once trained, the RNN-based controller gives the control inputs based on the current state and the history states.

The main contributions of this work can be summarized as follows. First, we propose novel, smooth counting quantitative semantics for STREL, which allows for optimizing spatial configurations. Second, we propose a hybrid optimization approach for solving the control synthesis problem for a multi-agent networked system from spatio-temporal specifications given as STREL formulas. 

Third, we provide fast, RNN-based real-time controllers for the control problem stated above. 
Fourth, while hybrid algorithms have been used before, to the best of our knowledge, this is the first time they have been employed for control synthesis. 

The rest of the paper is organized as follows. Notation and preliminaries are provided in Sec.~\ref{sec:prelim}. The control synthesis problem from STREL specifications is formulated in Sec.~\ref{sec:problem}. 
The new STREL quantitative semantics is described 
in Sec.~\ref{sec:newstrel}, and the corresponding optimization approach is discussed in Sec.~\ref{sec:solOPT}. We describe how the RNN-based controllers are learnt from data in Sec.~\ref{sec:learning_control}. The effectiveness of the overall proposed framework is demonstrated in a case study in Sec.~\ref{sec:CS}. We conclude with final remarks and directions for future work in Sec. \ref{sec:conclusion}.

\section{Preliminaries}
\label{sec:prelim}

\subsection{System Dynamics and Connectivity Conditions}
\label{sec:prelim_Dyn}

Consider a team of $N$ robotic agents labeled from the set $S = \{1,2,\ldots,N\}$. Each agent $l \in S$ has a state $x_l[k]= (q_l[k] , a_l)$ at (discrete) time $k$, where $q_l[k]\in \mathcal{Q} \subset \mathbb{R}^n$ 
is its dynamical state (e.g., position in space), and $a_l \in \mathcal{A}$ is an attribute that does not change over time (e.g., agent type), where $\mathcal{A}$ is a set of labels. The state of the team at time $k$ is denoted by $x[k]= [x_1[k]^T, \ldots, x_N[k]^T]^T$. We assume that the dynamics of each agent $l\in S$ is given by: 
\begin{equation}
\label{dyn}
q_{l}{[k+1]} = f(q_{l}[k],u_l[k]),
\end{equation}
where $u_l[k] \in \mathcal{U} \subset \mathbb{R}^m$ is the control input for agent $l$ at time $k$, $\mathcal{U}$ is the set of admissible controls, and $f:\mathbb{R}^n\times \mathbb{R}^m\rightarrow\mathbb{R}^n$. The control inputs of the team at time $k$ are denoted by $u[k]=[u_1[k]^T,\ldots,u_N[k]^T]^T$.

\begin{example} \label{ex:example_state}
 Consider a team of $N=8$ robotic agents (MANET\footnotemark ) moving in a $2D$ Euclidean space (see Fig.\ref{fig:example1}). The state of agent $l\in S=\{1,..,8\}$ at time $k$ is $x_l [k]= (q_l[k], a_l)$, where $q_l \in \mathbb{R}^2$ is the position of agent $l$, and $a_l\in \mathcal{A}=\{ {endDevice}, {coordinator},{router}\}$. Specifically, $a_3=a_5=a_6 =endDevice$, $a_1 =a_7=a_8 = coordinator$ and $a_4=a_2 = router$. 
 \demo
\end{example}

\footnotetext{A Mobile Ad-hoc sensor NETwork (MANET) is a team of robotic agents connected wirelessly. The agents are usually deployed to monitor environmental changes such as pollution, humidity, light and temperature. Each agent can be equipped with sensors, processors, radio transceivers, and batteries. It can move independently in any direction and change its connection links to the other agents. Moreover, the agents can be of different types and their behaviour and communication can depend on their types.}

\begin{figure}[hbt!]
\begin{center}
\includegraphics[width=.50\linewidth]{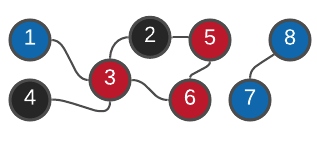}
\caption{A team of $N=8$ robotic agents with types {\color{red} $endDevice$ (red)}, {\color{black} $router$ (black)}, and {\color{blue} $coordinator$ (blue)}. The lines are edges in the connection graph.}
\label{fig:example1}
\end{center}
\end{figure}

For $H \in \mathbb{N}$, we use $\mathbf{x}_l^H$ to denote the \emph{state run} $x_l{[0]}x_l{[1]} \ldots x_l{[H]}$ of agent $l$, which is given by the dynamical state sequence $q_l{[0]}q_l{[1]} \ldots q_l{[H]}$ generated by the control sequence $u_l[0] \ldots u_l[H-1]$ starting at $q_l{[0]}$, along with the agent attribute $a_l$. We denote the state run of the team $x{[0]} x{[1]} \ldots x{[H]}$ by $\mathbf{x}^H$. Similarly, we use $\mathbf{u}^{H-1}$ to denote the control sequence of the team $u[0] \ldots u[H-1]$.

Two agents $l$ and $l'$ can communicate at time $k$ if they are connected according to given conditions. For instance, agents may be connected if the Euclidean distance between them is less than a ﬁxed communication range. Alternatively, agents may be connected according to a Delaunay triangulation or a Voronoi diagram (see Sec.~\ref{sec:CS}). We model the time-dependent inter-agent connectivity using an undirected graph (connection graph) $\lambda [k] = (S,E[k])$, where $E[k]\subseteq S\times S$. Specifically, $(l,l')\in E[k]$ means that $l$ and $l'$ are connected. We use ${\bm \lambda}^H$ to denote the sequence of connection graphs over time horizon $H$, i.e. $ {\bm \lambda}^H = \lambda[0]\lambda[1] \ldots \lambda[H]$. A \emph{route} $\tau =\tau_1 \tau_2 \ldots $ is a path on a connection graph with $\tau_i$ representing the label of the agent at position $i \in \mathbb{Z}_{\geq 0}$ on route $\tau$. The function ${\bm\tau}(l):S \rightarrow \mathbb{Z}_{>0} \cup \varnothing $ returns the position of agent $l$ on route $\tau$ if $l$ is on $\tau$ and returns $\varnothing$ otherwise.  Finally, the set of routes on the connection graph $\lambda{[k]}$ that start at agent $l$ is denoted by $Routes(\lambda{[k]},l)$ where $\forall \tau \in Routes(\lambda{[k]},l),\tau_1=l$. 

For notational simplicity, whenever there is no risk for confusion, we will drop the superscript $H$ from ${\bm \lambda}^{H}$, ${\bm x}^{H}$ and $x_l^{H}$ and use ${\bm \lambda}$, ${\bm x}$ and $x_l$, respectively. The horizon $H$ will be clear from context. We will use $[a,b]$ to denote the discrete time interval $[a,b] \cap \mathbb{Z}$ with $a,b \in \mathbb{Z}_{\geq 0}$.

\subsection{Spatio-Temporal Reach and Escape Logic (STREL)} \label{sec:prelim-strel}

STREL is a logic capable of describing complex behaviors in multi-agent teams. Formal definitions of the STREL syntax and semantics can be found in \cite{Bartocci2017}. Here, we give informal definitions for the syntax and the qualitative semantics. We provide, however, a formal definition of the quantitative semantics because we will refer to it later in the paper. Informally, STREL formulas are formed using atomic propositions $p$, which in this paper are attributes from the set $\mathcal{A}$, or predicates $\mu_{g(q_l[k])\sim r}$ defined over the dynamical states of the agents, where $g:\mathcal{Q}\rightarrow \mathbb{R}$, $r \in \mathbb{R}$, and 
$\sim \in \{\leq,>\}$; logical operators ($\neg,\vee, \wedge $); temporal operators: \emph{eventually} ($F_{[a,b]}$) and  \emph{always} ($G_{[a,b]}$) where $[a,b]$ is a discrete time interval with $a,b \in \mathbb{Z}_{\geq 0}$; and spatial operators: \emph{reach} ($\mathcal{R}_{\leq d}^f$), \emph{escape} ($\mathcal{E}^f_{>d}$) and \emph{surround} ($\odot^{f}_{\leq d}$), where $f$ is a distance function such as Euclidean distance $dist(l,l')$ or $hops(l,l')$, which returns the number of edges for the shortest route between $l,l'$, and $d$ is a scalar.

 The qualitative semantics define the satisfaction of a given STREL formula for $({\bm \lambda}, x_l)$ at time $k$. For instance, $({\bm \lambda}, x_l[k]) \models p$ ($\models$ reads ``satisfies") if the attribute of agent $l$ is $a_l = p$; $({\bm \lambda}, x_l[k]) \models \mu_{g(q_l[k])\sim r}$ if $g(q_l[k])\sim r$; $({\bm \lambda}, x_l[k]) \models F_{[a,b]} \varphi$ if $\exists k'\in [k+a,k+b]$ such that $({\bm \lambda}, x_l[k']) \models \varphi$; and $({\bm \lambda}, x_l[k]) \models G_{[a,b]} \varphi$ if  $\forall k'\in [k+a,k+b]:$  $({\bm \lambda}, x_l[k']) \models \varphi$. For the spatial operator \emph{reach}, $({\bm \lambda}, x_l[k]) \models \varphi_1 \mathcal{R}^{f}_{\leq d} \varphi_2$ if $\varphi_2$ is satisﬁed at agent $l'$ reachable from $l$ through a route $\tau$ with $\tau_1=l$ such that $f(l,l')\leq d$, and $\varphi_1$ is satisﬁed at $l$ and all the other agents between $l,l'$ on $\tau$ and we call such a route $\tau$ a \emph{satisfying route}. Similarly, $({\bm \lambda}, x_l[k]) \models \mathcal{E}^{f}_{> d} \varphi$ if there exists a route $\tau$ with $\tau_1=l$; and an agent $l'\in \tau$ such that $f(l,l')>d$, and $\varphi$ is satisﬁed at all agents $\tau_1\tau_2\ldots\tau_{{\bm\tau}(l')-1}$. The operator \emph{surround} expresses the notion of an agent with state that satisfies $\varphi_1$ being surrounded by agents with states satisfying $\varphi_2$ within a distance $f\leq d$. It can be derived from operators \emph{reach} and \emph{escape} as $\varphi_1 \odot^{f}_{\leq d} \varphi_2 = \varphi_1 \wedge \neg \left(\varphi_1 \mathcal{R}^{f}_{\leq d} \neg (\varphi_1 \vee \varphi_2)\right) \wedge \mathcal{E}^{f}_{ > d}  \varphi_1 \wedge \varphi_1 \mathcal{R}^{f}_{ \leq d} \varphi_2$.  
We note that we added the term ($\wedge \varphi_1 \mathcal{R}^{f}_{ d} \varphi_2$) to the definition provided in \cite{Bartocci2017} to avoid false satisfaction in the case of isolated agents that satisfy $\varphi_1$ while not being surrounded by agents satisfying $\varphi_2$ (see Ex.\ref{ex:reach_escape}). Additional operators can be defined or derived from the operators above and are not introduced here for brevity.

The \emph{time horizon} of a STREL formula $\varphi$ is denoted by $hrz(\varphi)$ and is defined as the smallest time point in the future for which the states of the agents are needed to determine the satisfaction of the formula. Hereafter, we assume that the horizon $H$ of $x_l^{H}$, ${\bm x}^{H}$ and ${\bm \lambda}^{H}$ is at least equal to $hrz(\varphi)$.

\begin{example} \label{ex:reach_escape} 
Consider the team of robotic agents from Ex.\ref{ex:example_state} (Fig.\ref{fig:example1}), the distance function $hops(l,l')$ and atomic propositions from the set of attributes $\mathcal{A}=\{endDevice,router,coordinator\}$. The formula $router \: \mathcal{R}^{hops}_{\leq 1} \: endDevice$ is satisﬁed at agent $2$ because there is a (satisfying) route $\tau = \tau_1 \tau_2$ with $\tau_1 = 2$ of type $router$ and  $\tau_2 = 3$ of type $endDevice$ with a distance of at most $1$ \emph{hop} from agent $2$. The formula $\mathcal{E}^{hops}_{> 2} \neg endDevice$ is satisﬁed at agent 5 because all of the agents in the route $\tau = \tau_1 \tau_2 \tau_3 = 5, 6, 3$ are $endDevices$. However, it is violated at agent $6$ because there satisfying routes with $hops(6,l')>2$. The formula $coordinator \odot^{hops}_{\leq 2} endDevice$ is satisfied at agent $1$ but violated at agents $7$ and $8$. Note that operator \emph{surround} as defined in \cite{Bartocci2017} would suggest that the same formula is satisfied at agents $7$ and $8$ despite them not being surrounded by agents of type $endDevice$.\demo
\end{example}

The quantitative valuation of a given STREL formula $\varphi$ is defined by a real-valued \emph{robustness function} $\rho_{_{}}$. The robustness $\rho_{_{}}$ of a STREL formula $\varphi$ with respect to $({\bm \lambda},x_l)$ at time $k$ is calculated recursively by \cite{Bartocci2017}:
\begin{subequations}\label{eq:strel-org}
\begin{gather}
        \rho_{_{}} ({\bm \lambda},x_l,p,k)  = \iota(p,x_l[k])\\
        \rho_{_{}} ({\bm \lambda},x_l,\mu,k)  = \iota(\mu,x_l[k])\\
        \rho_{_{}} ({\bm \lambda},x_l,\neg \varphi,k)  = - \rho_{_{}} ({\bm \lambda},x_l,\varphi,k)\\
        \begin{aligned}
        \rho_{_{}} ({\bm \lambda},x_l,\varphi_1 \vee \varphi_2 , k) = & \\ \max (\rho_{_{}} ({\bm \lambda},x_l,\varphi_1,&k), \rho_{_{}} ({\bm \lambda},x_l,\varphi_2,k))
        \end{aligned}\\
    \begin{aligned}
        \rho_{_{}} ({\bm \lambda},x_l,F_{[a,b]} \varphi,k) & = \\ \max_{k'\in [k+a,k+b]} \rho_{_{}}&({\bm \lambda},x_l,\varphi,k')
    \end{aligned}\\
    \begin{aligned}
        \rho_{_{}} ({\bm \lambda},x_l,G_{[a,b]} \varphi, k) & =\\ \min_{k'\in [k+a,k+b]} &\rho_{_{}}  ({\bm \lambda},x_l,\varphi,k')
    \end{aligned}\\
    \begin{aligned}
        \rho_{_{}} ({\bm \lambda},x_l,\varphi_1 \mathcal{R}_{\leq d}^f \varphi_2,k) & =\max_{\tau \in Routes(\lambda[k],l)}\max_{l'\in \tau:f(l,l')\leq d} \\
        \bigg[\min(\rho_{_{}}({\bm \lambda},x_{l'},\varphi_2,k)&;
        \min_{j<{\bm\tau}(l')}\rho_{_{}}({\bm \lambda},x_{\tau_j},\varphi_1,k))\bigg]
    \end{aligned}\\
    \begin{aligned}
        \rho_{_{}} ({\bm \lambda},x_l,\mathcal{E}^f_{>d} \varphi,k) & =
        \max_{\tau \in Routes(\lambda[k],l)}\max_{l'\in \tau:f(l,l')\leq d}\\ 
        & \qquad \min_{j<{\bm\tau}(l')}\rho_{_{}}\left({\bm \lambda},x_{\tau_j},\varphi,k\right)
    \end{aligned} \label{eq:strel-org-escape}
\end{gather}
\end{subequations}
where $\iota$ is the \emph{signal interpretation function} defined for atomic propositions and predicates by
\begin{equation}
\iota(p,x_l[k]) =\begin{cases}
		            	\rho_{max}, & \text{if $a_l=p$}\\
                     -\rho_{max}, & \text{otherwise}
		            \end{cases}\\
\end{equation}
\begin{equation}
\iota(\mu_{g(q_l[k])\sim r},x_l[k]) = \begin{cases}
			        g(q_l[k])-r, & \text{if $\sim =\: >$}\\
                     r-g(q_l[k]), & \text{if $\sim =\: \leq$}
		             \end{cases}
\end{equation}

\begin{theorem}\cite{Bartocci2017}
The STREL robustness defined by \eqref{eq:strel-org} is sound, i.e., positive robustness indicates satisfaction of the specification, and negative robustness indicates violation of the specification. 
\end{theorem}

\subsection{Smooth Robustness Approximation}
\label{sec:prelim-smooth}
The $\max$ and $\min$ functions can be approximated by \cite{Pant2017}:

\begin{equation*} 
\begin{array}{l}
\max(a_1,...,a_n)\approx {\widetilde{\max}}(a_1,...,a_n) = \frac{1}{\beta}
\ln \big (\sum_{i=1}^n e^{\beta a_i} \big),\\
\min(a_1,...,a_n) \approx {\widetilde{\min}}(a_1,...,a_n)= -{\widetilde{\max}}(-a_1,...,-a_n)
\end{array}
\end{equation*}
and the approximation error is bounded by $\epsilon_{\beta}$: 
\begin{equation}\label{eq:approx_error}
\begin{array}{l}
0\leq \max(a_1,..,a_n) - {\widetilde{\max}}(a_1,..,a_n)\leq \frac{\ln(n)}{\beta}= \epsilon_{\beta} \\ 
\end{array}
\end{equation} 
with the maximum error when $a_1=a_
2=\ldots =a_n$. The approximation error approaches $0$ when $\beta \rightarrow\infty$.

\section{Problem Formulation and Approach} \label{sec:problem}

Consider a team of $N$ agents labeled from the set $S = \{1,\ldots,N\}$ with dynamics \eqref{dyn} and a cost function $J(u[k],x[k+1])$, which is the cost of ending in state $x[k+1]$ by applying the control inputs $u[k]$ at time $k$. Assume that spatio-temporal requirements are specified by a STREL formula $\varphi$ interpreted over the state run of the team $\mathbf{x}^H$, and $H = hrz(\varphi)$ is the planning horizon. The goal is to find a control sequence for the multi-agent team that maximizes the robustness of the STREL formula and minimizes the cost function.

\begin{problem}\label{problem:main}[Control Synthesis] Given a multi-agent team with dynamics (\ref{dyn}), STREL formula $\varphi$, initial connection graph $\lambda [0]$, initial state of the system $x[0]$, planning horizon $H = hrz(\varphi)$, and cost function $J$; find an optimal control sequence $\mathbf{u}^{*H-1}$ that maximizes the robustness score and minimizes the cost, i.e.: 
\begin{align}\label{eq:problem1}
\begin{aligned}
  \mathbf{u}^{*H-1} & = \argmaxA_{\mathbf{u}^{H-1}}  \rho_c\left({\bm \lambda},{\bm x},\varphi,0 \right)
    - \gamma \sum_{k=0}^{{H-1}}J\left(u[k],x[k+1]\right) \\
        &  s.t. \\
                 & q_{l}{[k+1]} = f(q_l[k],u_l[k]),\quad \forall l \in S, \forall k \in [0,H-1]\\
                 & u_l{[k]} \in \mathcal{U},\qquad \qquad \qquad \forall l \in S,  \forall k \in [0,H-1]\\
\end{aligned}
\end{align}
where $\gamma >0$ is a trade-off coefficient and $\rho_c\left({\bm \lambda},{\bm x},\varphi,0 \right)$ is the robustness function for the team at time $0$.
\end{problem}

To solve Problem \ref{problem:main}, we first develop sound \emph{counting STREL quantitative semantics} with robustness $\rho_c$, which has spatial counting and allows to optimize the spatial configuration for connectivity (Sec.~\ref{sec:newstrel}). We use a smooth version of the robustness $\Tilde{\rho}_c$ in the objective function, and employ a combination of heuristic and gradient-based optimization algorithms (Sec.~\ref{sec:solOPT}) to find a control sequence that maximizes the robustness of the STREL formula and minimizes the cost function. In addition, as the execution time for the optimization might not meet real-time control requirements, we propose training a RNN to learn controllers from state-control trajectories generated by solving the optimization problem with different initializations (Sec.~\ref{sec:learning_control}). The trained RNN is then used to predict control inputs at each time.

\section{STREL Counting Quantitative Semantics}
\label{sec:newstrel}

To motivate the new proposed STREL quantitative semantics, we start by discussing the limitations of the original semantics defined by the robustness function in \eqref{eq:strel-org}.  

First, the robustness does not depend on the distance between the agents. Consider the multi-agent team from Ex.\ref{ex:reach_escape}. The STREL formula $endDevice \: \mathcal{R}^{hops}_{\leq 2} router$ will have the same robustness score for agents $3,5$ and $6$ even though their distances to the nearest $router$ (agent $2$) are not the same: $hops(3,2)=hops(5,2)=1<hops(6,2)=2$. In practice, connectivity between robotic agents in a networked system depends on the spatial configuration and it is often the case that the smaller the distance between agents, the better the connectivity.

Second, the robustness from \eqref{eq:strel-org} does not have a notion of spatial counting. Consider again the team from Ex.\ref{ex:reach_escape} and note that the formula $router \: \mathcal{R}^{hops}_{\leq 1} endDevice$ will have the same robustness score for agents $2$ and $4$, even though agent $2$ has two satisfying routes $2,3$ and $2,5$, compared to one satisfying route $4,3$ at agent $4$. In practice, it is beneficial to maximize the number of \emph{satisfying routes} for a given formula at a given agent. Take, for instance, the formula ($robot \odot^f_{\leq d} target$), maximizing the number of \emph{satisfying routes} results in maximizing the number of agents of type $robot$ surrounding $target$. 

Third, the STREL robustness from \eqref{eq:strel-org} is defined only at the level of individual agents. We need a way to compute the robustness score for the team, which takes into account the number of agents that satisfy/violate a given formula. 

The new STREL quantitative semantics addresses these limitations and differs from the original semantics defined in~\cite{Bartocci2017} in three ways. First, the proposed semantics depends on the distances among agents. Second, it performs spatial counting of satisfying/violating routes at individual agents. Third, by defining spatial counting for satisfying/violating agents, 
it captures the satisfaction of the specification by the team. 

\subsection{Optimizing the Spatial Configuration for Connectivity}

To optimize the spatial configuration of a multi-agent team for connectivity, we require the robustness score to depend on the distance between agents. To this goal, we introduce a function $\sigma_{dist}$, which depends on the distance between agents $f(l,l')$ and a scalar $d$. Specifically, $\sigma_{dist}$ is a sigmoid function that take values between $[-1,1]$ depending on the ratio $d_{norm}= \frac{f(l,l')}{d}$ (Fig.\ref{fig:sigma_dist}) and is defined for $f(l,l')\leq d$ and $f(l,l')>d$ as follows:
\begin{align} 
\sigma_{dist}^{\leq}(d_{norm}) &= -\tanh(k_{d}(d_{norm}-1)) \label{eq:sigma_dist1}\\
\sigma_{dist}^{>}(d_{norm})    &= \tanh(k_{d}(d_{norm}-1)) \label{eq:sigma_dist2}
\end{align}
where $k_{d}$ is a hyperparameter that determines how fast $\sigma_{dist}$ changes its value.  

When computing the robustness, we take $\min(\sigma_{dist},\rho_c)$ (see Sec.~\ref{sec:newstrel}). Notice that $\sigma_{dist}$ allows the robustness score to vary beyond the distance constant $d$ as defined by $f(l,l') \sim d,\: \sim =\{\leq,>\}$ as opposed to the original definition in \eqref{eq:strel-org}.

\subsection{Spatial Counting for Routes}\label{sec:sigma_routes}

Consider the robustness function of the spatial operator $escape$ defined by \eqref{eq:strel-org-escape} and define the robustness of a given route $\tau$ as $\rho_{\tau} := \max_{l'\in \tau:f(l,l')\leq d} \min_{j<{\bm\tau}(l')}\rho_{_{}}\left({\bm \lambda},x_{\tau_j},\varphi,k\right)$. Notice that  $\rho_{_{}}\left({\bm \lambda},x_l,\varphi,k\right) = \max_{\tau \in Routes(\lambda[k],l) } \rho_{\tau}, $, which means that it is enough to have one satisfying route $\tau \in Routes(\lambda[k],l)$ (with $\rho_{\tau}>0$) to satisfy formula $\varphi$ and the robustness score does not vary depending on the number of routes that satisfy/violate $\varphi$. We address this limitation by introducing an additional function $\sigma_{routes}$, which depends on the number of routes that satisfy/violate a given formula. 

Formally, let $R^+, R^- \in \mathbb{N}$ be the number of satisfying/ violating routes for a given spatial operator, respectively. We define the function $\sigma_{routes}$ (Fig.\ref{fig:sigma_routes}) 
\begin{multline} \label{eq:sigma_routes}
    \sigma_{routes}(R^+,R^-) = \max\big(\frac{1}{1+e^{k_{R} R^{-})}},\\ \frac{1}{1+e^{-k_{R} (R^{+}-R^{-}) }} \big)
\end{multline}
where $k_{R}$ is a hyperparameter that determines how fast $\sigma_{routes}$ changes its value. 

\begin{figure*}
     \centering
     \begin{subfigure}[b]{0.42\textwidth}
         \includegraphics[width=\textwidth]{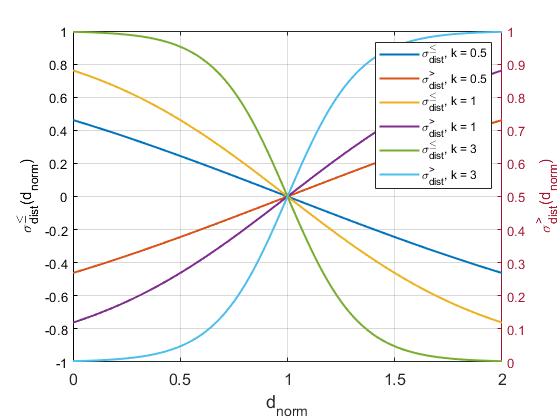}
    \caption{$\sigma_{dist}^{\leq}$ and $\sigma_{dist}^{>}$}
    \label{fig:sigma_dist}
     \end{subfigure}
     \hfill
     \begin{subfigure}[b]{0.42\textwidth}
         \includegraphics[width=\textwidth]{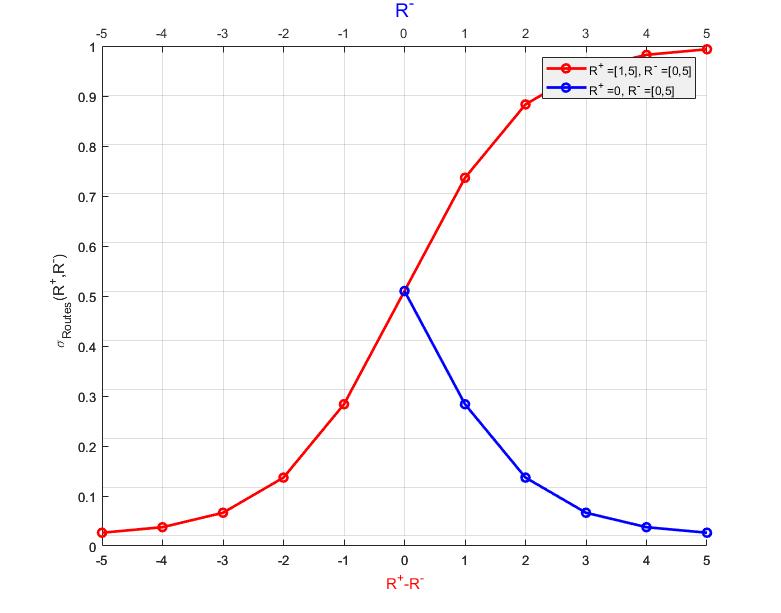}
    \caption{$\sigma_{routes}(R^+,R^-)$}
      \label{fig:sigma_routes}
     \end{subfigure}
    \caption{Behavior of the functions $\sigma_{dist}$ and $\sigma_{routes}$}
    \label{fig:sigma}
\end{figure*}

\subsection{Spatial Counting for Agents}

As mentioned above, the STREL semantics is defined at the level of individual agents. 
A na\"{i}ve method to compute the robustness of the team of $N$ agents is to consider the minimum of the robustness of individual agents. 
\begin{equation}\label{eq:rho_all_naive}
    \rho_{_{}}\left({\bm \lambda},{\bm x}, \varphi,k\right) = \min_{l \in S} \rho_{_{}}\left({\bm \lambda},x_l, \varphi,k\right)
\end{equation}
In this case, the robustness score for the team will reflect the worst robustness score among individual agents and will not depend on the number of agents that satisfy/violate the formula $\varphi$. Following a similar approach to route counting in Sec.~\ref{sec:sigma_routes}, we introduce $\sigma_{ag}(Ag^+,Ag^-)$, which allows for varying the robustness score depending on the number of agents that satisfy/violate the specification.

Let $Ag^{+},Ag^{-}$ be the number of agents that satisfy and violate the specification, respectively. Then $\sigma_{ag}$ is given by:
\begin{multline}\label{eq:sigma_agents}
    \sigma_{ag}(Ag^+,Ag^-) = \\ 
    \max\big( \frac{1}{1+e^{-k_{ag}Ag^{-}}},    \frac{1}{1+e^{k_{ag}(Ag^{-}-Ag^{+})}} \big)
\end{multline}

\subsection{Counting Robustness for STREL}

We are now ready to define the proposed counting robustness for the STREL spatial operators. 
\begin{multline*}
            \rho_c\left({\bm \lambda},x_l,\varphi_1 \mathcal{R}^{f}_{\leq d} \varphi_2,k\right)  =\\        \min\bigg[\sigma_{routes}(R^+,R^-) \max_{\tau \in Routes} \max_{l'\in \tau:dist(l')\leq d} \\\min\left(\rho_c\left({\bm \lambda},x_{l'},\varphi_2,k\right); \min_{j<{\bm\tau}(l')}\rho_c\left({\bm \lambda},x_{\tau_j},\varphi_1,k\right)\right)\\ ;\sigma_{dist}^{\leq}\left(d_{norm}\right)\bigg]
\end{multline*}
\begin{multline}\label{eq: strel-ro}
            \rho_c\left({\bm \lambda},x_l, \mathcal{E}^{f}_{> d} \varphi,k\right)  =  \min\bigg[ \sigma_{routes}(R^+,R^-) \\ \max_{\tau \in Routes} \max_{l'\in \tau:> d}\min_{j<{\bm\tau}(l')}\rho_c\left({\bm \lambda},x_{\tau_j},\varphi,k\right))\\ ; \sigma_{dist}^{>}\left(d_{norm}\right)\bigg]
\end{multline}

The robustness for the logical and temporal operators is the same as the one from \eqref{eq:strel-org}.

The robustness for the team at time $k$ is given by:  
\begin{multline} \label{eq: strel-ro-all}
    \rho_c({\bm \lambda},{\bm x},\varphi,k) = \sigma_{ag}(Ag^+,Ag^-) \\ \min_{l\in \{1,\ldots,N\}}(\rho_c({\bm \lambda},x_l,\varphi,k)).
\end{multline}

\begin{theorem}
The counting robustness of STREL defined by \eqref{eq: strel-ro} and \eqref{eq: strel-ro-all} is sound.
\end{theorem}
\begin{proof}[Proof]
[sketch] A formal proof is omitted due to space constraints.
Informally, soundness can be viewed as a sign consistency between the counting robustness $\rho_c$ and the original robustness in \cite{Bartocci2017} $\rho$. We show that the three functions $\sigma_{dist}$, $\sigma_{routes}$, $\sigma_{ag}$ introduced to $\rho$ do not affect the sign of the robustness, and thus show that $\rho_c$ is sound.

First, $\sigma_{routes}$ and $\sigma_{ag}$ are positive and are multiplied by the robustness function provided by the original semantics and, thus, do not change the sign of the robustness score.  
Second, $\sigma_{dist}$ changes in the range $[-1,1]$ and it is negative only when the distance predicate $f(l,l') \sim d$ is violated. Since we take the minimum between the robustness function and $\sigma_{dist}$, the robustness function will still give positive values for satisfaction and negative values for violation as before. Thus, the soundness of the proposed counting robustness follows from the soundness of the original robustness. 
\end{proof}

\section{Control Synthesis for STREL Specifications}
\label{sec:solOPT}

Problem \ref{problem:main} is a constrained non-linear optimization problem. The $max$ and $min$ functions in the semantics render the objective function non-differentiable. We modify \eqref{eq:problem1} by replacing the counting robustness $\rho_c$ with its smooth version of proposed smooth robustness $\Tilde{\rho_c}$ that we obtain by replacing the non-differentiable terms ($min/max$) with their smooth approximations described in Sec.~\ref{sec:prelim-smooth}.

We solve the new problem by employing a two-stage hybrid optimization method that utilizes a combination of heuristic and gradient-based optimization algorithms. In stage I, we explore the search space using a heuristic algorithm to find a good candidate solution. In stage II, we initialize a gradient-based algorithm with the best candidate solution found in stage I. Although many heuristic and gradient-based algorithms can be used for this approach, in this paper we use Particle Swarm Optimization (PSO)~\cite{kennedy1995particle} and Sequential Quadratic Programming (SQP)~\cite{polak2012optimization}, respectively. We name this two-stage optimization approach PSO+QSP. A quantitative comparison of the performance of SQP, PSO and PSO+SQP can be found in Sec.~\ref{sec:CS} (Tab.~\ref{tb:opt_CS}).

\section{Learning RNN-Based Controllers}
\label{sec:learning_control}

As already mentioned, solving the control synthesis problem by optimization can be expensive, and not feasible for real-time implementation. To address this, we propose to train a RNN using data obtained from off-line optimization, and then use it to generate the control at a given state. 

\textbf{Dataset generation.} Given a multi-agent team as described in Sec.~\ref{sec:prelim_Dyn}, a STREL formula $\varphi$, a planning horizon $H \geq hrz(\varphi)$, a set of $M$ initial team states $\{ x[0]^{[1]}, \ldots, x[0]^{[M]} \}$ and their corresponding initial communication graphs $\{ \lambda[0]^{[1]}, \ldots, \lambda[0]^{[M]} \}$, we generate a dataset $D$ by solving the control synthesis problem described in Sec.~\ref{sec:solOPT} and choosing $m \leq M$ state-control trajectories with robustness above a threshold $\epsilon_{min}$, i.e $D =\{\mathbf{x}_{[j]}^H,\mathbf{u}^{H[j]}|\Tilde{\rho_c}^{[j]}\geq \epsilon_{min} \}$, where $\epsilon_{min} \geq \epsilon_{\beta}$ is the robustness margin used to account for the approximation error $\epsilon_{\beta}$. 

\textbf{RNN implementation.} Due to the temporal operators, the satisfaction of STREL formulas is history-dependent. In other words, the control at each time step is, in general, dependent on the current state and past states $u[k]=g(x[0],..,x[k])$. For this reason, we choose Recurrent Neural Networks (RNN), which are neural networks with memory. To implement the RNN, we use a Long Short Term Memory (LSTM) Network~\cite{liu2020recurrent,Georgios}. LSTM has feedback channels and memory cells that can manage long-term dependencies by passing the history-dependence as hidden states. The function $g$ can be approximated as follows:
\begin{equation}\label{eq:RNN}
\begin{aligned}
    h{[k]}  &= \mathcal{R}(X[k],h[k-1],W_1)\\
    \hat{u}[k]    &= \mathcal{N}(h[k], W_2)
\end{aligned}
\end{equation}
where $W_1,W_2$ are the weight matrices of the RNN, $h[k]$ is the hidden state at time step $k$ and $\hat{u}[k]$ is the predicted control at $k$. The network is trained to minimize the error between the predicted control and the optimized control given in the dataset:
\begin{equation}
   \min_{W_1,W_2} \sum_{D} \sum_{k=0}^{H-1}\norm{ u[k]-\hat{u}[k]}^2.
\end{equation}

\begin{figure*}
     \centering
     \begin{subfigure}[b]{0.24\textwidth}
         \centering
         \includegraphics[width=\textwidth]{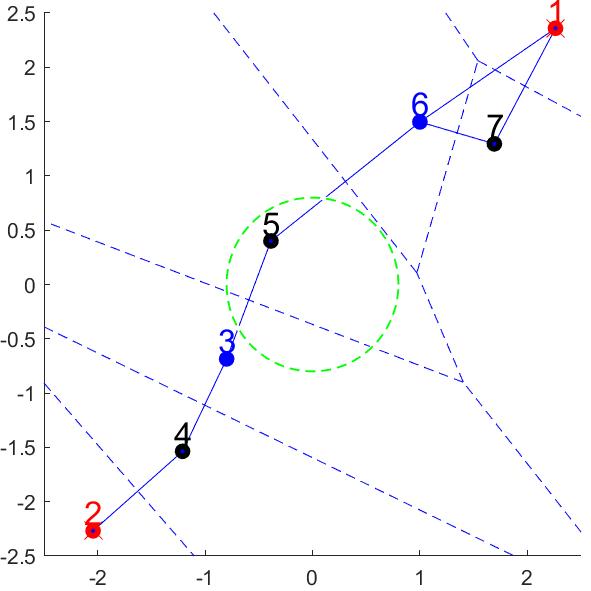}
         \caption{$k=0$}
         \label{fig:rnn0}
     \end{subfigure}
    \hfill
     \begin{subfigure}[b]{0.24\textwidth}
         \centering
         \includegraphics[width=\textwidth]{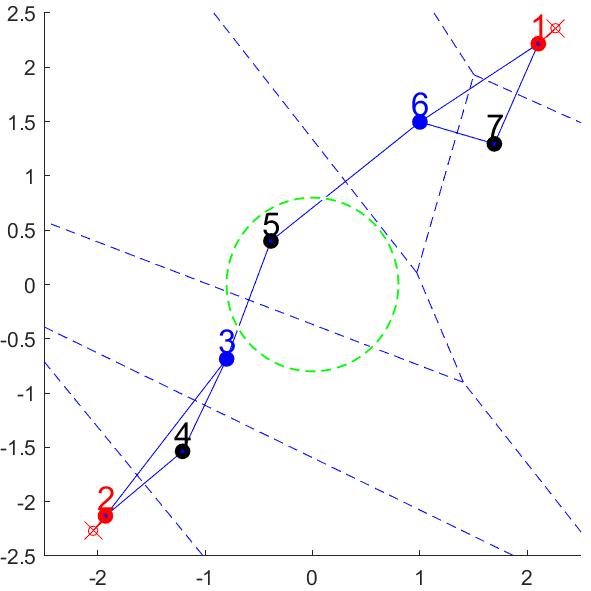}
         \caption{$k=1$}
         \label{fig:rnn3}
     \end{subfigure}
     \hfill
     \begin{subfigure}[b]{0.24\textwidth}
         \centering
         \includegraphics[width=\textwidth]{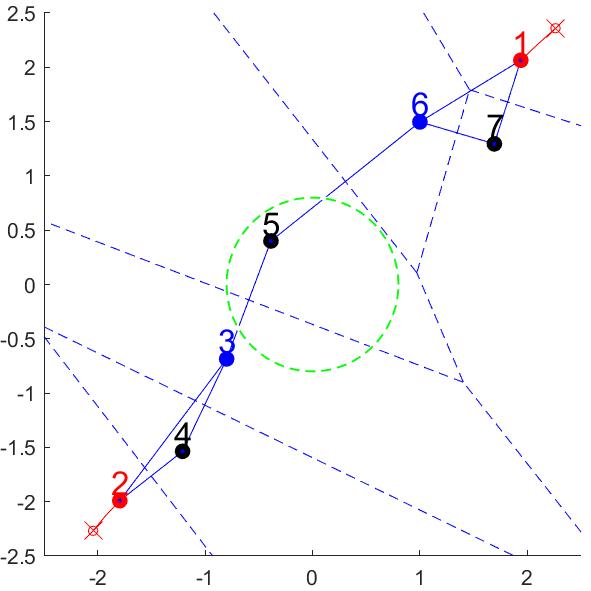}
         \caption{$k=2$}
         \label{fig:rnn5}
     \end{subfigure}
          \hfill
     \begin{subfigure}[b]{0.24\textwidth}
         \centering
         \includegraphics[width=\textwidth]{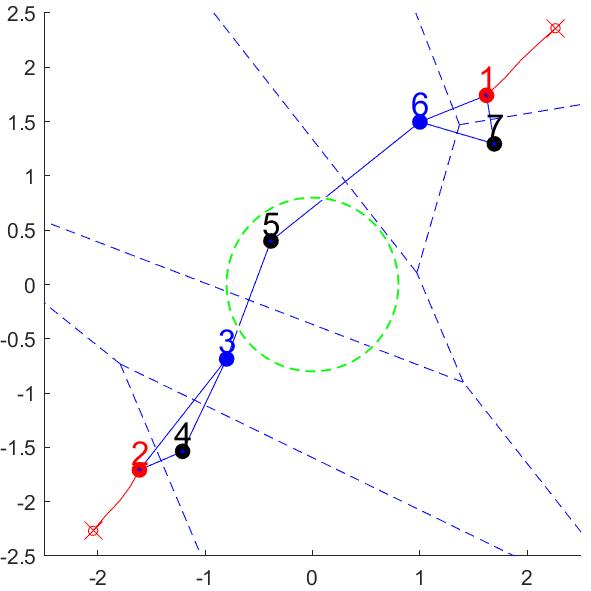}
         \caption{$k=4$}
         \label{fig:rnn7}
     \end{subfigure}
     \hfill
     \begin{subfigure}[b]{0.24\textwidth}
         \centering
         \includegraphics[width=\textwidth]{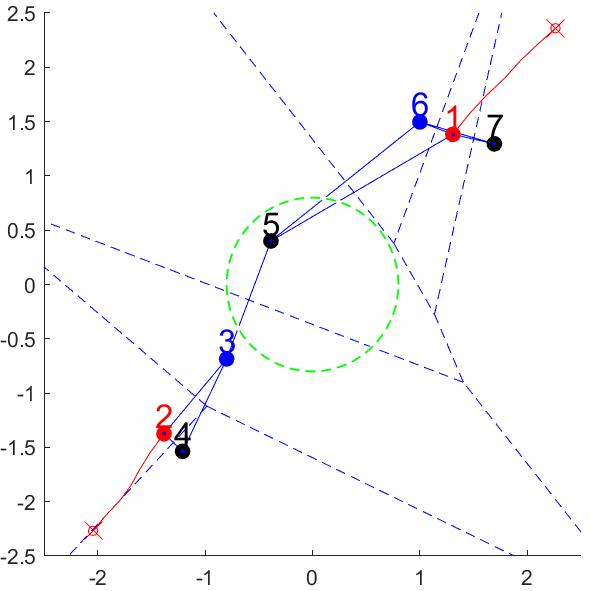}
         \caption{$k=6$}
         \label{fig:rnn8}
     \end{subfigure}
     \hfill
     \begin{subfigure}[b]{0.24\textwidth}
         \centering
         \includegraphics[width=\textwidth]{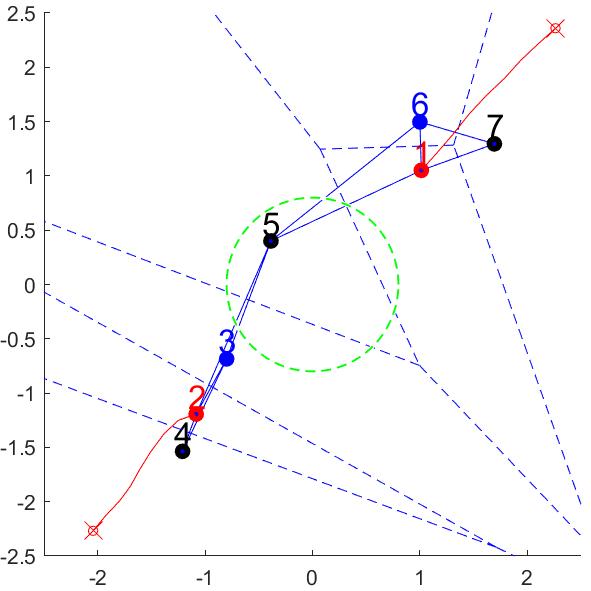}
         \caption{$k=8$}
         \label{fig:rnn10}
     \end{subfigure}
          \hfill
     \begin{subfigure}[b]{0.24\textwidth}
         \centering
         \includegraphics[width=\textwidth]{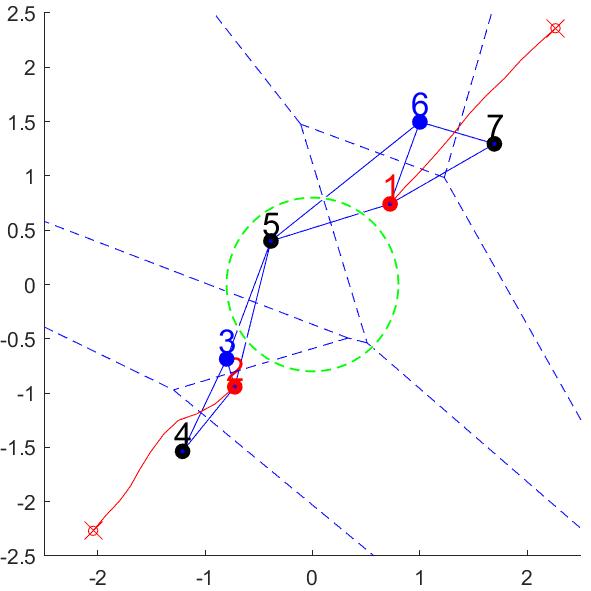}
         \caption{$k=10$}
         \label{fig:rnn11}
     \end{subfigure}
               \hfill
     \begin{subfigure}[b]{0.24\textwidth}
         \centering
         \includegraphics[width=\textwidth]{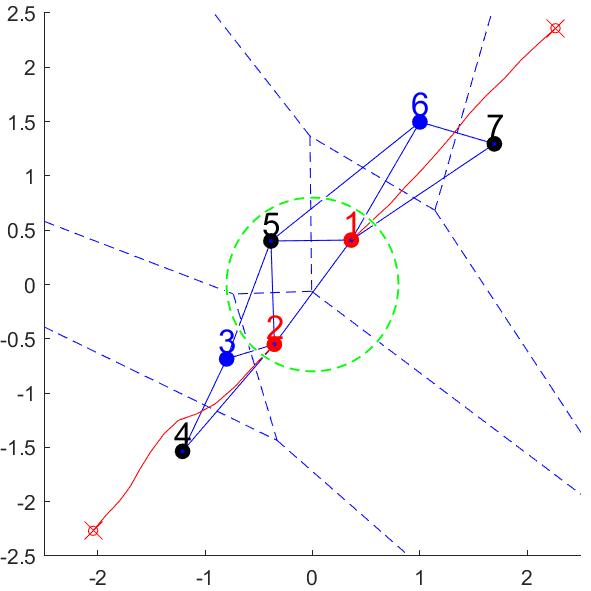}
         \caption{$k=12$}
         \label{fig:rnn13}
     \end{subfigure}
        \caption{Snapshots of the multi-agent team at different times. The agents are represented by labeled colored disks where labels are taken from the set $S=\{1,\ldots,7\}$ and colors correspond to types of agents: {\color{red} $endDevice$ (red)}{\color{black}, $router$ (black)}, and {\color{blue} $coordinator$ (blue)}. 
        The red solid lines represent the trajectories generated by agents $1$,$2$ using control inputs from the RNN-based controller. The solid and dashed blue lines represent the connection and Voronoi graphs of the team, respectively.}
        \label{fig:rnns}
\end{figure*}

\section{Case Study: Networked Robotic Agents}
\label{sec:CS}

In this section, we demonstrate the efficacy of our proposed framework with a case study. First, we solve the control synthesis problem for a multi-agent team and a given STREL formula off-line to generate state-control trajectories using the optimization approach described in Sec.~\ref{sec:solOPT}. Then, we use the state-control trajectories that satisfy the STREL formula to train a RNN to predict control inputs for the team at each time step (Sec. \ref{sec:learning_control}). We compare the trajectories obtained using RNN with those obtained by optimization. We also compare the results of optimization using different solvers, such as PSO, SQP, and PSO+SQP (Sec.~\ref{sec:solOPT}); and compare the performance of PSO under the original semantics \eqref{eq:strel-org}, \eqref{eq:rho_all_naive} and our proposed semantics \eqref{eq: strel-ro}, \eqref{eq: strel-ro-all}.

All the computations described in this section were performed on a PC with a Core i7 CPU @350GHz with 80 GB RAM. For optimization, we used a customized Particle Swarm Optimization (PSO) with $64$ particles and a MATLAB built-in implementation of SQP. To implement the RNN, we used the \emph{Python} package \emph{Pytorch}.

\textbf{System description} Consider a team of $N=7$ robotic agents (Fig.\ref{fig:rnns}) labeled from the set $S=\{1,2,\ldots,7\}$ in a 2D Euclidean space. The state of agent $l$ at time step $k$ is defined by $x_l[k]= (q_l[k],a_l)$ where $q_l[k] \in \mathbb{R}^2$ is the position of agent $l$ at time $k$ and $a_l \in \mathcal{A}= \{{endDevice}, {coordinator}, {router}\}$ is the type of the agent with $a_1=a_2=endDevice$, $a_4=a_5=a_7= router$ and $a_3=a_6=coordinator$. Agents $1$ and $2$ are controllable with dynamics given by
\begin{align} \label{eq:dyn_CS}
    q_l{[k+1]} &=q_l{[k]}+u_l{[k]},
\end{align}
where $l\in \{1,2\}$ and $u_l \in \mathcal{U} = [-0.2,0.2]$. 

\textbf{Connectivity conditions} Two agents $l$ and $l'$ are connected at time $k$ if both of the following conditions hold:
 \begin{itemize}
     \item The Euclidean distance between agents $l$ and $l'$  is less than a ﬁxed communication range $dist(l,l') \leq 2$.
     \item In the corresponding Voronoi diagram, the cells corresponding to agents $l$ and $l'$ are adjacent at time $k$.
 \end{itemize}
Snapshots of the team and the corresponding Voronoi and connection graphs at different times can be seen in Fig.\ref{fig:rnns}.

\textbf{Spatio-temporal specifications} We aim to steer agents $1$ and $2$ of type $endDevice$ from the initial positions to the area in the center (circle with radius $= 0.4$) in the time interval $[12,13]$, while staying connected to at least one agent of type $router$ at all times in $[3,13]$. We require all agents to avoid collision by keeping a safe distance of at least $0.15$ from each other at all times in $[0,13]$. The requirements described above can be specified by the following STREL formula, with $hrz(\varphi) =13$:

\begin{multline}\label{eq:phi_cs}
    \varphi = G_{[0,13]}(dist_{i\neq j}(q_i{[k]},q_j{[k]})>0.15)  \\
    \bigwedge G_{[3,13]} endDevice \mathcal{R}^{dist}_{\leq 2} router \\
    \bigwedge F_{[12,13]}(dist(q_{i,endDevice}{[k]}, origin)\leq 0.4).
\end{multline}

\textbf{Control problem} 
\emph{Given} the team of $N=7$ agents with dynamics \eqref{eq:dyn_CS}, STREL formula $\varphi$ (\eqref{eq:phi_cs}), initial connection graph $\lambda [0]$, initial state of the system $x[0]$, planning horizon $H = hrz(\varphi)=13$, 
\emph{find} an optimal control sequence $\mathbf{u}^{*12}$ that solves Problem \ref{problem:main}, where in \eqref{eq:problem1}, the cost function $J\left(u[k],x[k+1]\right)= \sum_{i=1}^{N} \sum_{k = 0}^{H-1} \norm{u_i[k]}^2$, $\gamma = 0.01$, and $\rho_c$ is replaced by its smooth version $\Tilde\rho_c$ (see Sec. \ref{sec:solOPT}).

\textbf{Comparing original and proposed semantics} We used PSO to solve the control synthesis problems with (i) the na\"{i}ve team robustness \eqref{eq:rho_all_naive} based on the original agent robustness \eqref{eq:strel-org} and (ii) the proposed robustness \eqref{eq: strel-ro},\eqref{eq: strel-ro-all} in the objective function. We solved each one with $100$ initializations and obtained no satisfying state-control trajectories for the original robustness and 69 satisfying state-control trajectories for the proposed robustness. We explain the results by noting that the proposed robustness has a varying search space (as opposed to the original robustness) due to the introduced functions $\sigma_{dist},\sigma_{routes},\sigma_{ag}$, which help the PSO algorithm to explore the search space and make it less prune to premature convergence.

\begin{table}[]
\begin{center}
\begin{tabular}{clcc}
\textbf{Algorithm} & \textbf{Success rate} & \textbf{\begin{tabular}[c]{@{}c@{}}Average \\ robustness\end{tabular}} & \textbf{\begin{tabular}[c]{@{}c@{}}Time\\ (seconds/run)\end{tabular}} \\
\textbf{SQP}       & 44.7                  & $0.0048$                                                               & $23.3$                                                        \\
\textbf{PSO}       & 71.7                  & $0.0037$                                                               & $30$                                                              \\
\textbf{PSO + SQP} & 93.8                  & $0.0050$                                                               & $32.9$    
\end{tabular}
\end{center}
\caption{Performance of different optimization methods}\label{tb:opt_CS}
\end{table}

\textbf{Dataset generation} We generated $m$ state-control trajectories with robustness score above a given threshold $\Tilde{\rho_c}^{*}\geq \epsilon_{min}=0.001>\epsilon_{\beta}$ by solving the control synthesis problem off-line. Choosing a good number of training samples $m$ is task-specific and depends on factors such as the complexity of the neural network and the complexity of the approximated function. In our case, we found that the RNN performs best with $m \geq 800$. We solved the control synthesis problem for the team with $1200$ random initializations. It took about \emph{6 hours} to execute the code and generate all the trajectories. We defined the \emph{success rate} as the percentage of state-control trajectories with robustness $\Tilde{\rho_c}^{*}\geq 0.001$. The success rate, average normalized robustness and computation times for solving the control synthesis problem with the proposed robustness using SQP, PSO and PSO+SQP are presented in Tab.~\ref{tb:opt_CS}.

\textbf{Training the RNN}
We used an LSTM network with four hidden layers to learn the controllers. Each hidden layer has $64$ nodes. We used $850$ trajectories for training and $275$ trajectories for testing. We trained the network for $700$ epochs. The training process took about six minutes.
\begin{figure}[]
  \begin{center}
    \includegraphics[width=0.30\textwidth]{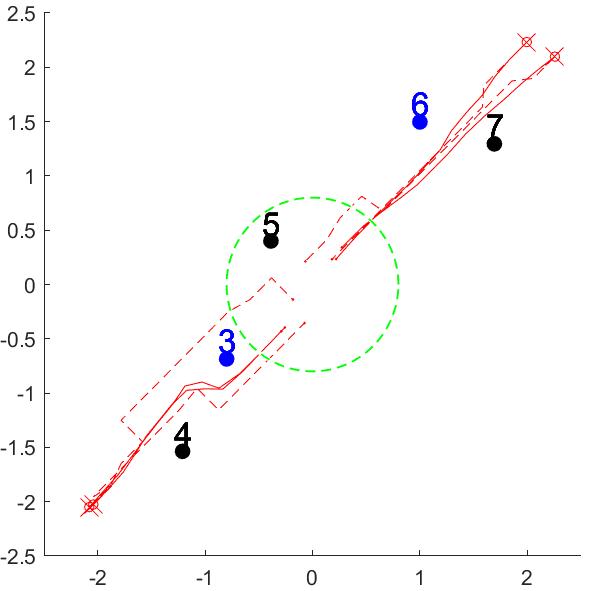}
  \caption{Trajectories generated from two different initializations (red crosses) of agents $1$ and $2$ (labels not shown for figure readability purpose)  using RNN-based controller (solid red) and optimization PSO+SQP (dashed red)}
    \label{fig:CS_optVSrnn}
  \end{center}
\end{figure}

\textbf{Results} The average robustness for the trajectories generated using the RNN-based controller from $275$ new initializations is $0.0037$ compared to $0.0050$ for trajectories generated using the proposed optimization method PSO+SQP. The average execution times for generating one trajectory using the RNN-based controller and PSO+SQP are $0.002$ seconds and $32.9$ seconds, respectively. The success rate for trajectories generated from new initializations using the RNN-based controller is $93\%$. Fig.\ref{fig:CS_optVSrnn} shows sample trajectories generated from two different initializations using the RNN-based controller and using SQP+PSO. The results presented above demonstrate that the learned RNN-based controllers achieve a high success rate and have significantly lower computation time compared to the PSO+SQP optimization.

\section{Conclusions and Future Research}
\label{sec:conclusion}

We proposed a framework for solving control synthesis problems for multi-agent networked systems from spatio-temporal specifications. We introduced new counting quantitative semantics for the Spatio-Temporal Reach Escape Logic (STREL), and used it to map the control problems to optimization problems. The proposed semantics are sound, smooth, and allow for spatial counting and optimizing the spatial configuration of the team for agents connectivity. We solve the optimization problem using a combination of heuristic and gradient-based algorithms in two stages. In the first stage, we utilize Particle Swarm Optimization (PSO) for search-space exploration to find the best candidate solution. In the second stage, Sequential Quadratic Programming (SQP) is initialized by the candidate solution obtained from PSO and employed until a pre-defined stopping criterion is met. To meet real-time control requirements, we learn Recurrent Neural Network (RNN) - based controllers from state-control trajectories generated by solving the optimization problem with different initializations. Directions of future research include using reinforcement learning for online control and extending the proposed framework to complex and large teams of agents. 

\bibliographystyle{IEEEconf}
\bibliography{main}

\end{document}